\documentclass[letterpaper, 10 pt, conference]{ieeeconf}  %

\IEEEoverridecommandlockouts                              %
\usepackage{tikz}
\usepackage{subcaption}
\usepackage{booktabs}
\usepackage{url}
\usepackage[hidelinks]{hyperref}
\usepackage{cite}
\usepackage{amsmath,amssymb,amsfonts}
\usepackage{algorithmic}
\usepackage{graphicx}
\usepackage{textcomp}
\usepackage{xcolor}
\def\BibTeX{{\rm B\kern-.05em{\sc i\kern-.025em b}\kern-.08em
    T\kern-.1667em\lower.7ex\hbox{E}\kern-.125emX}}
\newtheorem{thm}{Theorem }%
\newtheorem{proposition}{Proposition}%
\newtheorem{lemma}{Lemma}%
\newtheorem{defn}{Definition}%
\newtheorem{rem}{Remark }%
\newtheorem{example}{Example}[section]

\usepackage[ruled,vlined,linesnumbered]{algorithm2e}
\usepackage{makecell}

\usepackage{enumerate}

\newcommand{\set}[1]{\left\{#1\right\}}

\newcommand{\ra}{\rightarrow}
\newcommand{\Real}{\mathbb{R}}
\newcommand{\eps}{\varepsilon}

\renewcommand{\subset}{\subseteq}

\newcommand{\C}{\mathcal{C}}

\newcommand{\X}{\mathcal{X}}
\newcommand{\U}{\mathcal{U}}
\newcommand{\Ubox}{\mathcal{U}_{\text{box}}}
\newcommand{\Uball}{\mathcal{U}_{\text{ball}}}

\title{\LARGE \bf
Formal Verification of Control Lyapunov-Barrier Functions\\ for Safe Stabilization with Bounded Controls
}

\author{Jun Liu  %
\thanks{This research was supported in part by the Natural Sciences and Engineering Research Council of Canada and the Canada Research Chairs Program. The author also acknowledges the Digital Research Alliance of Canada and the Mathematics Faculty Computing Facility (MFCF) at the University of Waterloo for computing support.}%
\thanks{Jun Liu is with the Department of Applied Mathematics, University of Waterloo, Waterloo, Ontario N2L 3G1, Canada.  Email: \texttt{j.liu@uwaterloo.ca (Jun Liu)}
        }%
}

\begin{document}

\maketitle
\thispagestyle{empty}
\pagestyle{empty}

\begin{abstract}
We present verifiable conditions for synthesizing a single smooth Lyapunov function that certifies both asymptotic stability and safety under bounded controls. These sufficient conditions ensure the strict compatibility of a control barrier function (CBF) and a control Lyapunov function (CLF) on the exact safe set certified by the barrier. An explicit smooth control Lyapunov-barrier function (CLBF) is then constructed via a patching formula that is provably correct by design. Two examples illustrate the computational procedure, showing that the proposed approach is less conservative than sum-of-squares (SOS)-based compatible CBF-CLF designs.
\end{abstract}

\section{Introduction}

Ensuring both safety and stability is a central challenge in modern control design. While classical control theory focuses primarily on stabilization, many real-world systems, such as autonomous vehicles \cite{seo2022safety}, robotic systems \cite{li2024stabilization}, and industrial processes \cite{wu2020control}, must also operate safely within prescribed state and input limits. The need for controllers that enforce safety without compromising stability has therefore become a key research direction in recent years.

Control Lyapunov functions (CLFs) \cite{artstein1983stabilization,sontag1983lyapunov,sontag1989universal} offer a principled way to certify stability, while control barrier functions (CBFs) \cite{ames2016control,wieland2007constructive,xu2015robustness} provide formal safety guarantees by maintaining forward invariance of safe sets. Combining these two constructs has led to a growing body of work on safe stabilization and reach-avoid control; see, e.g., \cite{romdlony2016stabilization,meng2022smooth,mestres2025conversetheoremscertificatessafety,li2024stabilization,dawson2022safe,li2023graphical,ong2019universal,mestres2022optimization,quartz2025converse}. However, the coexistence of safety and stability constraints introduces a fundamental tension: a CLF that drives the system toward an equilibrium may require control actions that violate the CBF condition, and vice versa. Optimization-based methods, such as quadratic program (QP) formulations, have become popular for handling this trade-off in real time, yet their feasibility critically depends on the compatibility  of the CBF-CLF pair \cite{mestres2022optimization}. When such compatibility fails, controllers often rely on heuristic relaxations that weaken provable guarantees.

An alternative paradigm is to construct a single smooth function that simultaneously encodes both safety and stability, known as a control Lyapunov-barrier function (CLBF)~\cite{romdlony2016stabilization}. This unification removes the need for online optimization and provides explicit certificates and control laws for safe stabilization. Despite its appeal, constructing such a function is nontrivial. Existing CLBF formulations~\cite{romdlony2016stabilization} typically rely on strong structural assumptions about the safe set which may not be satisfiable~\cite{braun2017existence,braun2020comment}. 

It is therefore of interest to determine under what conditions a single smooth CLBF can exist. The work in~\cite{meng2022smooth} investigated this question through converse Lyapunov theorems, following the idea introduced in~\cite{liu2020converse}, and characterized robust safe stability using Lyapunov functions. Building on~\cite{meng2022smooth}, subsequent results in  \cite{mestres2025conversetheoremscertificatessafety} broadened the scope by deriving necessary conditions for the existence of CLBFs and for the compatibility of CBF-CLF pairs. Complementary to these developments,~\cite{quartz2025converse} established that a strictly compatible pair of control Lyapunov and control barrier functions exists if and only if there is a single smooth Lyapunov function that simultaneously certifies asymptotic stability and safety. Inspired by this line of work,~\cite{liu2025computing} introduced a computational framework that unifies strictly compatible CBFs and CLFs into a single CLBF for safe stabilization, but without explicitly accounting for input constraints. 

In this work, we extend the results of~\cite{liu2025computing} and develop verifiable conditions for synthesizing a single smooth Lyapunov function that certifies both asymptotic stability and safety under bounded controls. These conditions are derived using reformulations of Farkas' lemma to encode compatibility constraints under two types of input bounds, norm-bounded and hyperbox constraints, which can be readily verified using satisfiability modulo theories (SMT) solvers~\cite{gao2013dreal}. Once verified, a provably correct smooth control Lyapunov-barrier function (CLBF) can be constructed following~\cite{liu2025computing}. The resulting CLBF enables the use of the universal formulas of Lin and Sontag~\cite{lin1991universal} and Leyva \emph{et al.}~\cite{leyva2013global} to compute safe stabilizing controllers subject to input constraints. We demonstrate the computational procedure on two nonlinear examples, showing that the proposed method is less conservative than sum-of-squares (SOS)-based compatible CBF-CLF designs.

\section{Preliminaries and problem formulation}

\subsection{System description}

We consider a nonlinear control-affine system: 
\begin{equation}
    \label{eq:sys}
    \dot{x} = f(x) + g(x)u,
\end{equation}
where \( f:\mathbb{R}^n \to \mathbb{R}^n \) and \( g:\mathbb{R}^n \to \mathbb{R}^{n \times m} \) satisfy \( f(0) = 0 \).  
The objective is to design a feedback control law \( u = \kappa(x) \), with \( \kappa:\mathbb{R}^n \to \mathbb{R}^m \) and \( \kappa(0) = 0 \), such that the closed-loop dynamics  
\begin{equation}
    \label{eq:clsys}
    \dot{x} = f(x) + g(x)\kappa(x),
\end{equation}
preserve the origin as an equilibrium point and its solutions meet the following requirements:
\begin{enumerate}[(i)]
    \item for every initial state \(x(0)\in\mathcal{X}\subseteq\mathbb{R}^n\), the trajectory \(x(t)\) is well-defined and remains in \(\mathcal{X}\) for all \(t\ge 0\);
    \item the input satisfies \(\kappa(x(t))\in\mathcal{U}\subseteq\mathbb{R}^m\) for all \(t\ge 0\); 
    \item the origin is asymptotically stable.
\end{enumerate}

In this work, we focus on achieving (i)-(iii) under prescribed sets \(\mathcal{X}\) and \(\mathcal{U}\), which are described below.

\subsection{State constraints}

We consider multiple state constraints defined by  
\begin{equation}
\label{eq:constraint}
\mathcal{C}_i = \left\{ x \in \mathbb{R}^n \mid h_i(x) \le 1 \right\}, \quad i = 1, \ldots, N,
\end{equation}
where each function \( h_i:\mathbb{R}^n \to \mathbb{R} \) is continuously differentiable.  
The overall safe region is specified as the intersection of these individual constraint sets as 
\(
\mathcal{X}  = \bigcap_{i=1}^N \mathcal{C}_i. 
\)
Equivalently, the safe set can be characterized through the pointwise maximum function as
\begin{equation}
\label{eq:hmax}
\mathcal{X} = \left\{ x \in \mathbb{R}^n \mid h_{\max}(x) \le 1 \right\},
\end{equation}
with \( h_{\max}(x) := \max_{1 \le i \le N} h_i(x) \).  
Note that \( h_{\max} \) is generally nonsmooth and fails to be continuously differentiable in regions where multiple \( h_i \) attain the same value. We assume that the origin lies in the interior of $\mathcal X$. 

\subsection{Input constraints}

We consider two types of input constraints. The first is specified by a 2-norm bound of the form
\begin{equation}
    \label{eq:unorm}
    \mathcal{U}_{\text{ball}} = \left\{ u \in \mathbb{R}^m \mid \|u\|_2 \le \bar{u} \right\},
\end{equation}
where \(\bar{u} > 0\) denotes the maximum allowable control magnitude.

The second type is a hyperbox constraint, defined as
\begin{equation}
    \label{eq:ubox}
    \mathcal{U}_{\text{box}} = \left\{ u \in \mathbb{R}^m \mid \underline{u}_j \le u_j \le \overline{u}_j,\; j = 1, \ldots, m \right\},
\end{equation}
where \(\underline{u}_j, \overline{u}_j \in \mathbb{R}\), \(j = 1, \ldots, m\), denote the lower and upper bounds on each control component. We consider either $\U=\Uball$ or $\Ubox$ as the admissible control domain.

\subsection{Problem formulation}\label{sec:problem}

We adopt a unified CLF framework for safe stabilization under bounded inputs. Specifically, we aim to construct a continuously differentiable function \( W:\mathbb{R}^n \to \mathbb{R} \) that is positive definite, $W(0)=0$, and satisfies 
    \begin{equation}
        \inf_{u \in \U} \left[ L_f W(x) + L_g W(x)u \right] < 0,  
        \quad \forall x \in \mathcal{C} \setminus \{0\},
    \end{equation}
    where \( L_f W = \nabla W^\top f \), \( L_g W = \nabla W^\top g \), and 
    \begin{equation}
        \label{eq:set_C}
        \mathcal{C} = \left\{ x \in \mathbb{R}^n \mid W(x) \le 1 \right\}\subset \mathcal{X}.
    \end{equation}
Ideally, \(\mathcal{C}\) should provide a close inner approximation of the safe set \(\mathcal{X}\).  
Once such a function \(W\) is available, universal formulas of Lin and Sontag~\cite{lin1991universal} for norm-bounded inputs and Leyva \emph{et al.}~\cite{leyva2013global} for hyperbox constraints can be applied directly to \(W\) to yield a safe, stabilizing controller.  This highlights one of the main advantages of employing a single smooth Lyapunov function to certify both stability and safety.

\section{Formal verification of CLBFs with bounded controls}

In this section, we present the main results of the paper. We adopt the guaranteed patching approach~\cite{liu2025computing} and establish formally verifiable conditions that enable the provably correct construction of a smooth control Lyapunov-barrier function for safe stabilization under bounded controls. 

\subsection{Verification of strict CBFs with bounded controls}

Consider a continuously differentiable function \( h:\mathbb{R}^n \to \mathbb{R} \) defining the set
\begin{align}
\mathcal{C} & = \left\{ x \in \mathbb{R}^n \mid h(x) \le 1 \right\}, \label{safe-set}\\
\partial \mathcal{C} & = \left\{ x \in \mathbb{R}^n \mid h(x) = 1 \right\}.
\end{align}
We seek to verify that \(h\) is a \emph{strict} CBF on \(\mathcal{C}\), i.e., for every \(x \in \partial \mathcal{C}\), 
\begin{equation}
    \label{eq:cbf}
    \inf_{u \in \U} \left[ L_f h(x) + L_g h(x) u \right] < 0,
\end{equation}
where $\U=\Uball$ or $\Ubox$, defined in (\ref{eq:unorm}) and (\ref{eq:ubox}), respectively.

\begin{proposition}\label{prop:cbf}
The following hold:
\begin{enumerate}
\item If $\mathcal{U}=\mathcal{U}_{\text{ball}}$, the CBF condition \eqref{eq:cbf} is equivalent to
    \begin{equation}
        \label{eq:cbf_ball}
        h(x)=1 \;\Longrightarrow\; L_f h(x) \;<\; \bar{u}\,\|L_g h(x)\|_2 .
    \end{equation}

\item If $\mathcal{U}=\mathcal{U}_{\text{box}}$, the CBF condition \eqref{eq:cbf} is equivalent to
\begin{equation}
    \label{eq:cbf_box}
\begin{aligned}
    &h(x)=1 \;\Longrightarrow\;\\
    &
    L_f h(x)
    + \sum_{j=1}^m \!\left(m_j\,(L_g h(x))_j
        - r_j\,\big|(L_g h(x))_j\big|
    \right) < 0,    
\end{aligned}
\end{equation}
where $m_j=\tfrac{\underline{u}_j+\overline{u}_j}{2}$ and $r_j=\tfrac{\overline{u}_j-\underline{u}_j}{2}$.
\end{enumerate}
\end{proposition}

The proof is elementary and provided for completeness.

\begin{proof}
Fix $x$ with $h(x)=1$ and set $c:=L_f h(x)$ and $q:=L_g h(x)\in\mathbb{R}^{1\times m}$.

\emph{(1) Ball constraint.}
\[
\inf_{\|u\|_2\le \bar{u}} \big(c+ q u\big) < 0
\;\Longleftrightarrow\;
c + \inf_{\|u\|_2\le \bar{u}} \langle q, u\rangle < 0.
\]
Hence \eqref{eq:cbf_ball} follows from 
$\inf_{\|u\|_2\le \bar{u}} \langle q, u\rangle = -\bar{u}\|q\|_2$.

\emph{(2) Box constraint.}
\[
c + \inf_{u\in\mathcal{U}_{\text{box}}} \sum_{j=1}^m q_j u_j < 0,
\quad
\mathcal{U}_{\text{box}}=\prod_{j=1}^m [\underline u_j,\overline u_j].
\]
Separability of the box constraints in $u$ yields
\[
\inf_{u\in\mathcal{U}_{\text{box}}} \sum_{j=1}^m q_j u_j
= \sum_{j=1}^m \inf_{\underline{u}_j\le u_j\le \overline{u}_j} q_j u_j
= \sum_{j=1}^m \min\{q_j\underline{u}_j,\;q_j\overline{u}_j\}.
\]
Using $\min\{a,b\}=\tfrac{a+b}{2}-\tfrac{|a-b|}{2}$ with $a=q_j\underline u_j$, $b=q_j\overline u_j$,
\[
\min\{q_j\underline{u}_j,q_j\overline{u}_j\}
= \tfrac{\underline u_j+\overline u_j}{2}\,q_j - \tfrac{\overline u_j-\underline u_j}{2}\,|q_j|
= m_j q_j - r_j |q_j|.
\]
Thus $c+\sum_{j=1}^m (m_j q_j - r_j |q_j|)<0$, which is \eqref{eq:cbf_box}.
\end{proof}

Over a compact domain $\X$, both conditions (\ref{eq:cbf_ball}) and (\ref{eq:cbf_box}) can be readily verified by a $\delta$-complete SMT solver such as dReal \cite{gao2013dreal}. 

\subsection{Verification of strictly compatible CBF-CLF pairs with bounded controls}

Let $h:\mathbb{R}^n \to \mathbb{R}$ be a candidate CBF and $V:\mathbb{R}^n \to \mathbb{R}$ a candidate CLF.   
Assume that both $h$ and $V$ are continuously differentiable, $V$ is positive definite, and the set $\mathcal{C}$ defined in~\eqref{safe-set} contains the origin in its interior.  
We now introduce a strict compatibility condition between $h$ and $V$.

\begin{defn}[\cite{quartz2025converse}]\label{def:compatibility}
    We say that $h$ and $V$ are strictly compatible if the following conditions hold: 
    \begin{enumerate}
        \item For every $x \in \mathcal{C} \setminus \{0\}$, there exists $u\in\U$ such that  
    \begin{equation}\label{eq:CLFV}
        L_f V(x) + L_g V(x) u < 0.
    \end{equation}
        \item For every $x \in \partial \C$, there exists $u\in\U$ such that 
    \begin{equation}\label{eq:CLFBoundary}
        L_f V(x) + L_g V(x) u < 0,
    \end{equation}
    and
    \begin{equation}\label{eq:CBFBoundary}
        L_f h(x)+L_g h(x) u < 0.
    \end{equation}
    \end{enumerate}
\end{defn}

\subsubsection{Verification of strict CLF conditions with bounded controls}

The compatibility conditions above include a CLF condition on \(V\), i.e., (\ref{eq:CLFV}), which we address first.  
The following proposition is essentially the same as Proposition~\ref{prop:cbf}, with \(h\) replaced by \(V\) and \(\partial\mathcal{C}\) replaced by \(\mathcal{C}\setminus\{0\}\).

\begin{proposition}\label{prop:clf}
    The following hold:
\begin{enumerate}
\item If $\mathcal{U}=\mathcal{U}_{\text{ball}}$, the CLF condition \eqref{eq:CLFV} is equivalent to
    \begin{equation}
        \label{eq:clf_ball}
        (h(x)\le 1 \wedge x\neq 0)\;\Longrightarrow\; L_f V(x) \;<\; \bar{u}\,\|L_g V(x)\|_2 .
    \end{equation}

\item If $\mathcal{U}=\mathcal{U}_{\text{box}}$, the CLF condition \eqref{eq:CLFV} is equivalent to
\begin{equation}
    \label{eq:clf_box}
\begin{aligned}
    &(h(x)\le 1 \wedge x\neq 0) \;\Longrightarrow\;\\
    &
    L_f V(x)
    + \sum_{j=1}^m \!\left(m_j\,(L_g V(x))_j
        - r_j\,\big|(L_g V(x))_j\big|
    \right) < 0,    
\end{aligned}
\end{equation}
where $m_j=\tfrac{\underline{u}_j+\overline{u}_j}{2}$ and $r_j=\tfrac{\overline{u}_j-\underline{u}_j}{2}$.
\end{enumerate}
\end{proposition}

\begin{rem}
When using a $\delta$-complete SMT solver such as dReal~\cite{gao2013dreal}, spurious counterexamples may arise near the origin when verifying (\ref{eq:clf_ball}) and (\ref{eq:clf_box}).  
To mitigate this, we follow the approach of~\cite{liu2025formally,liu2025physics}, which isolates the near-origin region and verifies that the dominant homogeneous parts of the CLF and the dynamics ensure local stabilization.  
For linearly stabilizable dynamics and a quadratic CLF $V(x)=x^\top P x$, this reduces to verifying that $V$ serves as a valid Lyapunov function for the closed-loop system under a linear feedback controller $u=Kx$, as detailed in~\cite{liu2025physics,liu2025formally}.  
We adopt the same approach here, with the additional step of performing the verification on a sufficiently small sublevel set of $V$ to ensure that the control input $u=Kx$ remains within the bounds imposed by $\Uball$ or $\Ubox$.  
This verified inner level set is then used as a base region for subsequent CLF verification under bounded controls, by checking (\ref{eq:clf_ball}) and (\ref{eq:clf_box}) outside this base region.
\end{rem}

\subsubsection{Verification of strict CLF-CBF compatibility with bounded controls}

We now focus on the strict compatibility conditions~\eqref{eq:CLFBoundary} and~\eqref{eq:CBFBoundary} under bounded controls. Because these conditions involve the quantifier structure $\forall x\,\exists u$, we employ variants of Farkas' lemma to obtain quantifier-free reformulations that can be handled by SMT solvers such as dReal~\cite{gao2013dreal}. Depending whether $\U=\Ubox$ or $\Uball$, we use either the linear or conic version of Farkas' lemma.

\begin{lemma}[Farkas' Lemma]\label{lem:strictfarkas}
Let $A \in \Real^{n \times m}$ and $b \in \Real^n$.  
The system $A u < b$, $u \in \Real^m$, has a solution if and only if every nonnegative 
$\lambda \in \Real^n$ with $\lambda^\top A = 0^\top \in \Real^{1 \times m}$ and  
$\sum_{i=1}^n \lambda_i=1$ also satisfies $\lambda^\top b > 0$.
\end{lemma}

Here we consider a strict inequality formulation and add a normalization to $\lambda$ so that the equivalent condition can be readily verified by SMT solvers such as dReal \cite{gao2013dreal}. A proof of this variant of Farkas' Lemma can be found in \cite{liu2025computing}. Furthermore, to cope with norm-bounded controls, we use the following lemma. The proof can be found in the Appendix. %

\begin{lemma}[Second-Order Conic Farkas' Lemma]\label{lem:conic_farkas}
Let $A \in \Real^{n \times m}$, $b \in \Real^n$, and $B>0$.  
The system $A u < b$ has a solution satisfying $\|u\|_2 \le B$ if and only if for every $\lambda \in \Real^n_{\ge 0}$, $\tau \in \Real_{\ge 0}$, and $y \in \Real^m$ satisfying 
$y + A^\top \lambda = 0$, $\|y\|_2 \le \tau$, and $\sum_{i=1}^n \lambda_i + \tau = 1$ also satisfies $\lambda^\top b + B \tau > 0.$
\end{lemma}

Based on Lemmas~\ref{lem:strictfarkas} and~\ref{lem:conic_farkas}, we derive the following equivalent formulations of the compatibility conditions~\eqref{eq:CLFBoundary} and~\eqref{eq:CBFBoundary} under bounded controls. These formulations can be directly verified using an SMT solver.

\begin{proposition}
\label{prop:clbf-compat-bounded-control}
The following hold:
\begin{enumerate}[(i)]
\item If $\U=\Uball$, the strict compatibility conditions \eqref{eq:CLFBoundary}--\eqref{eq:CBFBoundary} are equivalent to 
\begin{align}
&((h(x)=1) \notag\\
&\wedge (\lambda_1,\lambda_2,\tau \ge 0,\; \lambda_1+\lambda_2+\tau=1)  \notag\\
&\wedge (\,L_g V(x)^\top \lambda_1 + L_g h(x)^\top \lambda_2 + y = 0,\; \|y\|_2 \le \tau\,))  \notag\\
&\Longrightarrow\;
\lambda_1 L_f V(x) + \lambda_2 L_f h(x) < \bar{u}\,\tau. 
\label{eq:farkas-ball}
\end{align}

\item If $\U=\Ubox$, the strict compatibility conditions \eqref{eq:CLFBoundary}--\eqref{eq:CBFBoundary} are equivalent to 
\begin{equation}
\label{eq:farkas-box}
\begin{aligned}
&\big(h(x)=1 \\
&\ \wedge\ \left\{
\begin{aligned}
&\lambda_1,\lambda_2,\ \lambda_j^\pm \ge 0 \quad (j=1,\dots,m),\\
&\lambda_1+\lambda_2+\sum_{j=1}^m(\lambda_j^+ + \lambda_j^-) = 1
\end{aligned}\right. \\
&\ \wedge\ \lambda_1 L_g V(x)+\lambda_2 L_g h(x)
          +\sum_{j=1}^m(\lambda_j^+ - \lambda_j^-)\,e_j = 0\big)\\
&\Longrightarrow\;
\lambda_1 L_f V(x)+\lambda_2 L_f h(x)
   < \sum_{j=1}^m(\lambda_j^+\overline u_j-\lambda_j^-\underline u_j).
\end{aligned}
\end{equation} 
Here \( e_j \in \mathbb{R}^{1 \times m} \) denotes the \( j \)-th standard basis row vector. 
\end{enumerate}
\end{proposition}

\begin{proof}
Fix $x$ with $h(x)=1$ and set
\[
A(x):=\begin{bmatrix}L_g V(x)\\[2pt] L_g h(x)\end{bmatrix}\in\Real^{2\times m},\qquad
b(x):=\begin{bmatrix}-L_f V(x)\\[2pt]-L_f h(x)\end{bmatrix}\in\Real^{2}.
\]

\noindent\textbf{(i) $\U=\Uball$.}
Strict compatibility means
\[
\exists\,u:\ A(x)u<b(x)\quad\text{and}\quad \|u\|_2\le\bar u.
\]
By Lemma~\ref{lem:conic_farkas} with $B=\bar u$, feasibility is equivalent to: for all
$\lambda=(\lambda_1,\lambda_2)^\top\!\ge0$, $\tau\ge0$, and $y\in\Real^m$ with
$y+A(x)^\top\lambda=0$, $\|y\|_2\le\tau$, $\lambda_1+\lambda_2+\tau=1$, one has
$\lambda^\top b(x)+\bar u\,\tau>0$, which is \eqref{eq:farkas-ball}.

\noindent\textbf{(ii) $\U=\Ubox$.} The feasibility problem can be written as
\[
\begin{bmatrix}
L_g V(x) \\[2pt]
L_g h(x) \\[2pt]
I_m \\[2pt]
-I_m
\end{bmatrix} u
<
\begin{bmatrix}
-\,L_f V(x) \\[2pt]
-\,L_f h(x) \\[2pt]
\overline u \\[2pt]
-\,\underline u
\end{bmatrix},
\qquad u\in\Real^m.
\]
Denote this compactly as $\tilde A(x)u<\tilde b(x)$.
By Lemma~\ref{lem:strictfarkas}, this system is feasible if and only if 
every nonnegative multiplier $\tilde\lambda\ge0$ satisfying 
$\tilde\lambda^\top\tilde A(x)=0$ and $\mathbf{1}^\top\tilde\lambda=1$ 
also satisfies $\tilde\lambda^\top\tilde b(x)>0$. Partition $\tilde\lambda$ as
\[
\tilde\lambda = (\lambda_1,\lambda_2,\lambda^+_1,\ldots,\lambda^+_m,\lambda^-_1,\ldots,\lambda^-_m)^\top. 
\]
The stationarity condition $\tilde\lambda^\top\tilde A(x)=0$ becomes
\[
\lambda_1 L_g V(x)+\lambda_2 L_g h(x)
   +\sum_{j=1}^m(\lambda^+_j-\lambda^-_j)e_j=0,
\]
and the dual inequality $\tilde\lambda^\top\tilde b(x)>0$ reads
\[
-\lambda_1 L_f V(x)-\lambda_2 L_f h(x)
   +\sum_{j=1}^m(\lambda^+_j\,\overline u_j-\lambda^-_j\,\underline u_j)>0,
\]
which rearranges directly to~\eqref{eq:farkas-box}.
\end{proof}

\subsection{Guaranteed patching of compatible CLF-CBF}

Theoretically, it has been shown that the existence of a strictly compatible CBF-CLF pair is equivalent to the existence of a single smooth control Lyapunov-barrier function (CLBF) (see~\cite{quartz2025converse}; see also~\cite{liu2025computing} for an alternative proof). We restate the following result from~\cite{quartz2025converse,liu2025computing}.

\begin{thm}[\cite{quartz2025converse,liu2025computing}]\label{thm:patch}
    Let $h$ and $V$ be strictly compatible as defined in Definition \ref{def:compatibility}. Suppose that the set $\C$ is compact. Then there exists a continuously differentiable function $W:\,\Real^n\ra\Real$ with the following properties:
    \begin{enumerate}
        \item $\mathcal{C}=\set{x\in\mathcal \Real^n \mid W(x)\le 1}$.
        \item For every $x \in \mathcal{C} \setminus \{0\}$, there exists $u\in\U$ such that  
    \begin{equation}\label{eq:CLFW}
        L_f W(x) + L_g W(x) u < 0.
    \end{equation}
    \end{enumerate}
\end{thm}

In particular, we assume that \(\mathcal{C}\) is compact and follow the construction in~\cite{liu2025computing} to explicitly define \(W\) from \(V\) and \(h\).  
Fix \(\varepsilon \in (0,1)\) and define the inner boundary band
\[
\partial\mathcal{C}_{\varepsilon}^{-} := \{x \mid 1-\varepsilon \le h(x) \le 1\}.
\]
Strict compatibility of \(V\) and \(h\) on \(\partial\mathcal{C}\), together with compactness of \(\mathcal{C}\), implies that there exists \(\varepsilon > 0\) such that compatibility also holds on \(\partial\mathcal{C}_{\varepsilon}^{-}\). We can determine such an $\eps>0$ by an SMT solver \cite{liu2025computing}, by replacing $h(x)=1$ in (\ref{eq:farkas-ball}) and (\ref{eq:farkas-box}) with $1-\eps\le h(x)\le 1$. Moreover, without loss of generality, we may assume
\[
V(x) \le 1 \quad \forall x \in \mathcal{C},
\qquad
V(x) \le h(x) \quad \forall x \in \partial\mathcal{C}_{\varepsilon}^{-}.
\]
Otherwise, \(V\) can be rescaled by the factor
\begin{equation}\label{eq:scale}
\alpha = \frac{1-\varepsilon}{\max_{x\in \mathcal{C}} V(x)},
\end{equation}
as shown in~\cite{liu2025computing}.  
Define the smooth bump function
\[
b=
\begin{cases}
\exp\!\Bigl(-\dfrac{1}{\varepsilon^2 - (h(x)-1)^2} + \dfrac{1}{\varepsilon^2}\Bigr), &
1-\varepsilon < h(x) < 1,\\[4pt]
1, & h(x) \ge 1,\\
0, & h(x) \le 1-\varepsilon.
\end{cases}
\]
It is shown in~\cite{liu2025computing} that
\begin{equation}\label{eq:patched_W}
W(x) := (1-b(x))\,V(x) + b(x)\,h(x)
\end{equation}
satisfies the conclusions of Theorem~\ref{thm:patch} by construction.

\begin{rem}[Softmax relaxation of the safe set]
Following~\cite{liu2025computing}, we can approximate $h_{\max}$ by the smooth log-sum-exp function with temperature parameter $\tau>0$:
\begin{equation}
\label{eq:softmax}
h_{\mathrm{sm}}(x;\tau)\;:=\;\frac{1}{\tau}\,\log\!\Big(\sum_{i=1}^N e^{\tau\,h_i(x)}\Big).
\end{equation}
As discussed in \cite{liu2025computing}, we have
\[
\C:=\{h_{\mathrm{sm}}\le 1\}
\;\subseteq\;
\{h_{\max}\le 1\}
\;=\;
\mathcal{X}.
\]
In other words, the $1$-sublevel set of $h_{\mathrm{sm}}$ provides a guaranteed under-approximation of the safe set $\X$, and as $\tau\to\infty$, $\C$ can get arbitrarily close to $\X$. In the numerical examples, we simply set \( h(x) = h_{\mathrm{sm}}(x;\tau) \) for a fixed \(\tau > 0\) and verify, using Proposition~\ref{prop:cbf}, that \(h\) serves as a valid CBF under bounded controls. We then proceed to verify that there exists a Lyapunov function $V$ compatible with $h$ using Propositions  \ref{prop:clf} and \ref{prop:clbf-compat-bounded-control}. Once verified, the Lyapunov function $W$ defined in (\ref{eq:patched_W}) serves both as a certificate and as a tool for computing provably safe, stabilizing controllers. 
\end{rem}

\begin{rem}[Universal formulas for safe stabilization]
To extract a controller using \( W \) that satisfies the conclusions of Theorem~\ref{thm:patch}, we can directly apply the formula in~\cite{lin1991universal} for controls with a 2-norm bound and the formula in~\cite{leyva2013global} for controls with a hyperrectangle bound. Both formulas, by construction, yield a safe and stabilizing controller.
\end{rem}

\begin{figure}[h!]
    \centering
    \includegraphics[width=\linewidth]{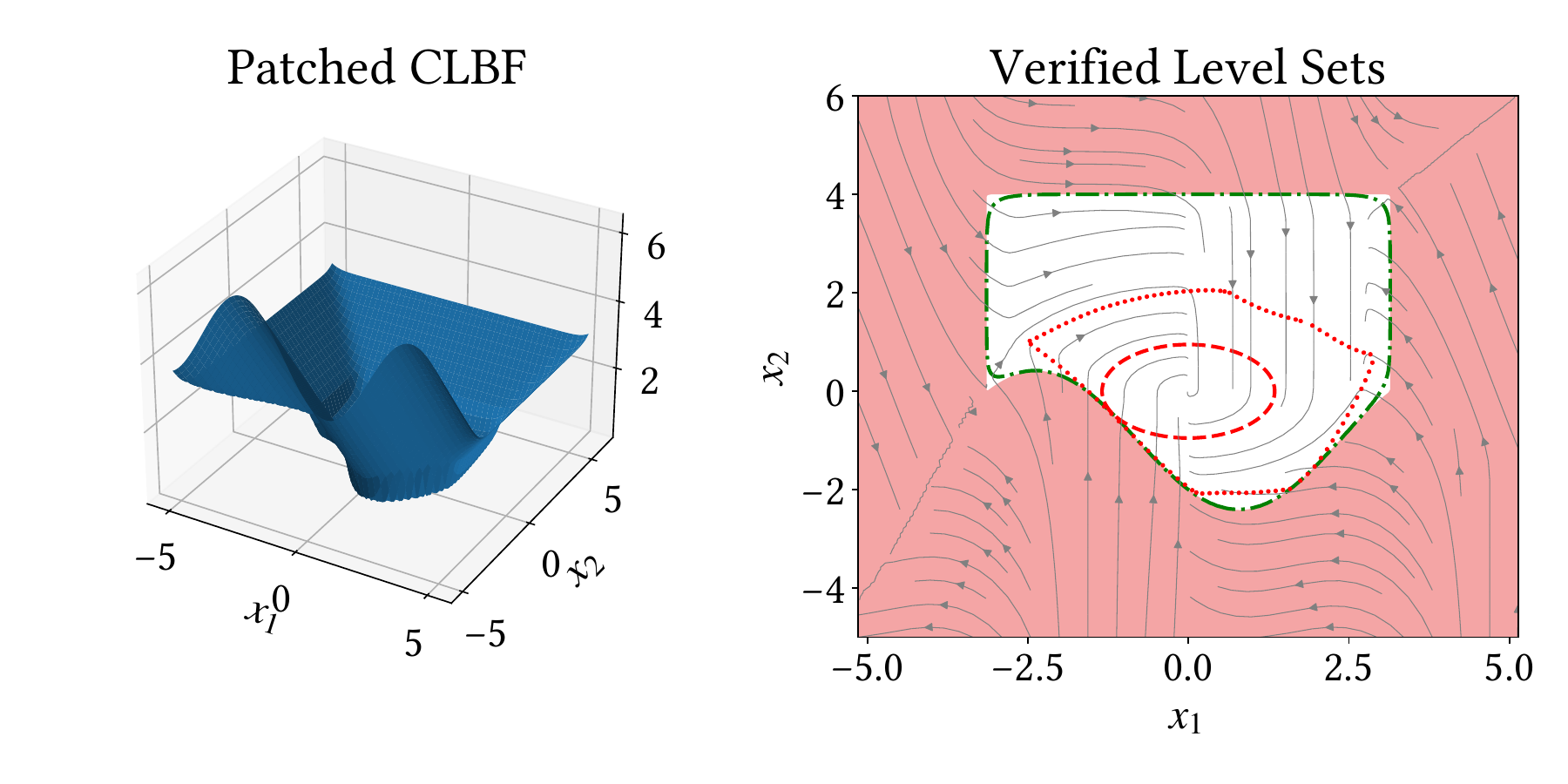}
    \caption{Formally verified smooth CLBF constructed by patching a strictly compatible CBF-CLF pair for Example~\ref{ex:2d_toy_verified}. The green dash-dotted curve marks the verified safe stabilization region under bounded inputs $u \in [-1,1]$, shown relative to the unsafe region (light red). For comparison, the dotted and dashed red curves indicate the largest verified regions from the SOS-based compatible CBF-CLF approach~\cite{dai2024verification} and a quadratic CLF, respectively.}
    \label{fig:2d_toy_verified}
\end{figure}

\begin{figure*}[h!]
    \centering
    \includegraphics[width=0.3\linewidth]{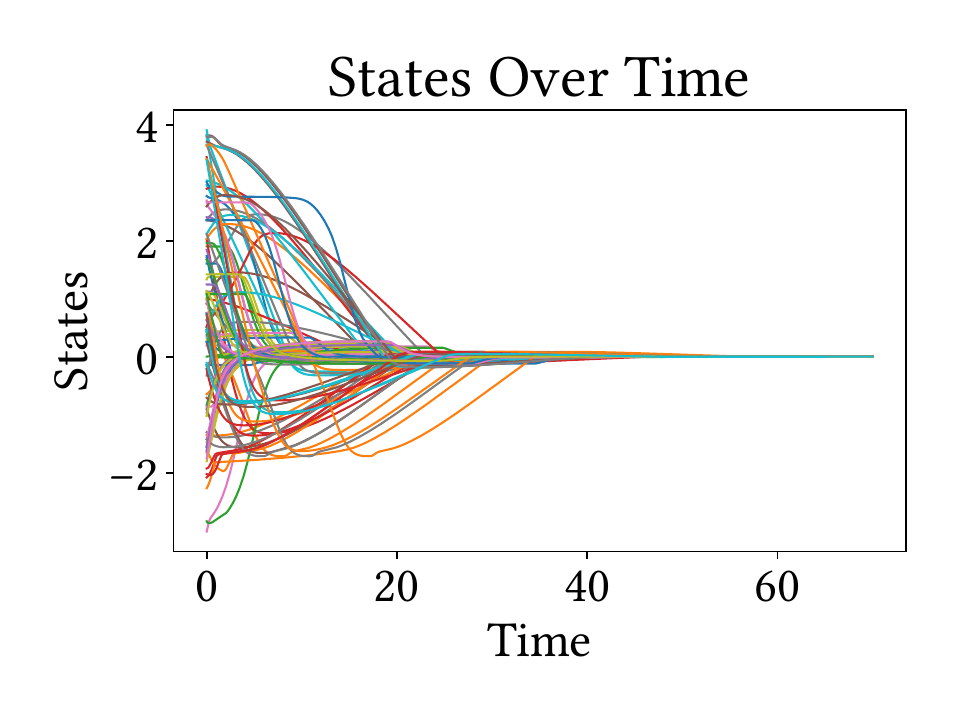}
    \includegraphics[width=0.3\linewidth]{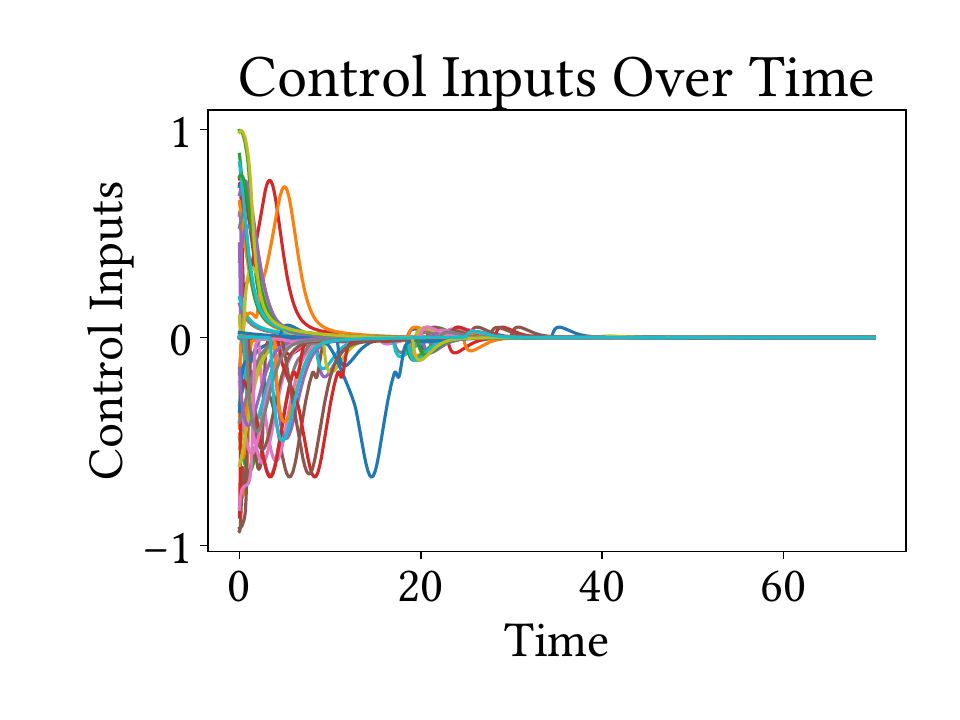}
    \includegraphics[width=0.3\linewidth]{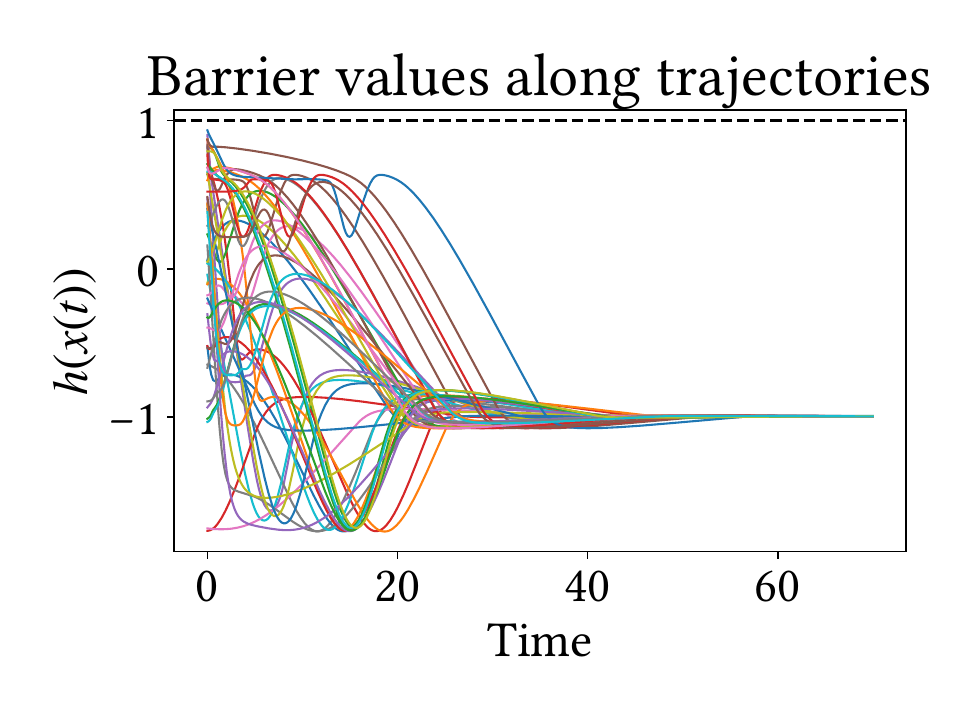}
    \caption{Simulated trajectories, control inputs, and corresponding barrier function evaluations for Example~\ref{ex:2d_toy_verified}. All trajectories converge, control input bounds are satisfied, and the barrier function values stay below the safety threshold of~1.}
    \label{fig:2d_toy_verified_trajectories}
\end{figure*}

\section{Numerical examples}\label{sec:examples}

We demonstrate the proposed method on two nonlinear control examples. All experiments were conducted on a server node equipped with an Intel(R) Xeon(R) Gold 6326 CPU @ 2.90,GHz (16 cores, 2 threads per core). The problems were solved using the LyZNet toolbox~\cite{liu2024tool}, with dReal~\cite{gao2013dreal} serving as the verification backend. The code is available at \url{https://git.uwaterloo.ca/hybrid-systems-lab/lyznet} under \texttt{examples/clbf-bounded-control}.

\begin{example}\label{ex:2d_toy_verified}
We consider the following toy example taken from~\cite{dai2024verification}: 
\[
    f(x) = \begin{bmatrix} 0 \\ -\sin x_1 \end{bmatrix},
    \quad
    g(x) = \begin{bmatrix} 1 \\ -1 \end{bmatrix},
\]
The state domain is set to $[-\pi,\pi] \times [-3,4]$, and the input domain to $\U=[-1,1]$. Following~\cite{liu2025computing}, we construct a smooth \emph{log-sum-exp} barrier function $h = h_{\mathrm{sm}}(x;\tau)$ as defined in~(\ref{eq:softmax}) with $\tau = 4.5$. Using the verifiable conditions established in Propositions \ref{prop:cbf}--\ref{prop:clbf-compat-bounded-control}, we formally verify that $h$ is a strict CBF and is strictly compatible with the quadratic CLF $V(x) = x_1^2 + 2x_2^2$ on the band $\partial\mathcal{C}_{\varepsilon}^{-}=\{x : 1 - \varepsilon \le h(x) \le 1\}$ with $\varepsilon = 0.5$ and a scaling factor $0.012814877271582478$, as defined in~(\ref{eq:scale}). Consequently, the function $W$ defined in~(\ref{eq:patched_W}) constitutes a valid control Lyapunov-barrier function. 

Figure~\ref{fig:2d_toy_verified} illustrates the patched CLBF and the phase portrait of the closed-loop system under the Lin-Sontag controller~\cite{lin1991universal}. Figure~\ref{fig:2d_toy_verified_trajectories} presents 50 simulated trajectories starting from the set $\mathcal{C} = \{h \le 1\}$, along with the evaluation of $h$ to demonstrate safety. The barrier function $h$ alone cannot serve as a Lyapunov function on $\mathcal{C}$, not only because it attains negative values but also because it has a stationary point at $x = (0.1294085,\, 0.94176161)^\top$ (numerically verified via gradient descent). Patching it with the quadratic CLF produces a provable CLBF, as demonstrated here. By construction, the control input constraints are satisfied. 
\end{example}

\begin{figure}[!ht]
    \centering
    \includegraphics[width=\linewidth]{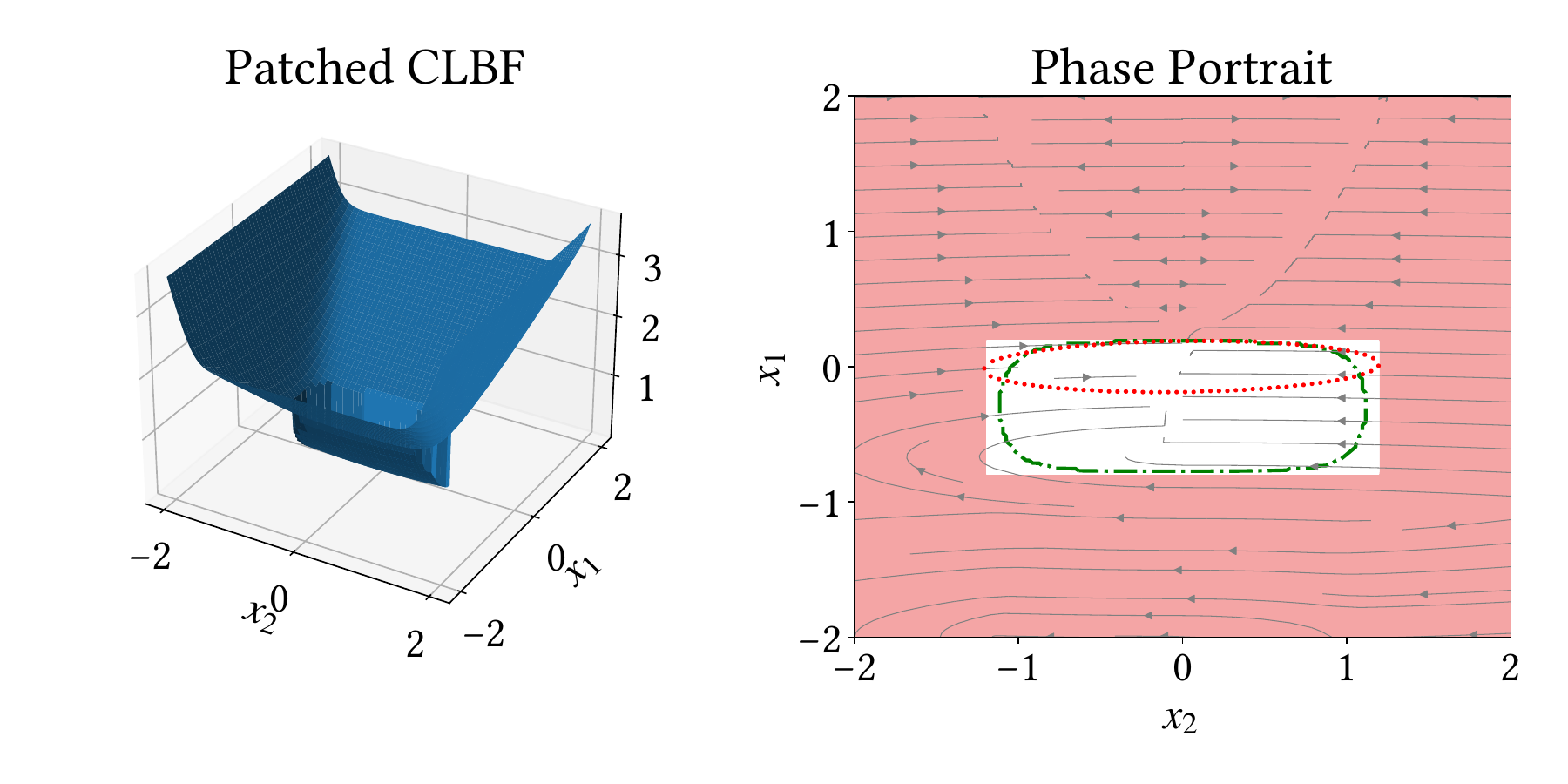}
    \caption{A formally verified CLBF for Example~\ref{ex:3d_power} subject to input constraint $\U=[-2,2]^2$. The green dash-dot curve represents the formally verified safe stabilization region, relative to the unsafe region (shaded in light red). For comparison, the best reported SOS CBF-CLF result \cite{dai2024verificationv1} is shown in dotted red.}
    \label{fig:3d_power}
\end{figure}

\begin{figure*}[!ht]
    \centering
    \includegraphics[width=0.3\linewidth]{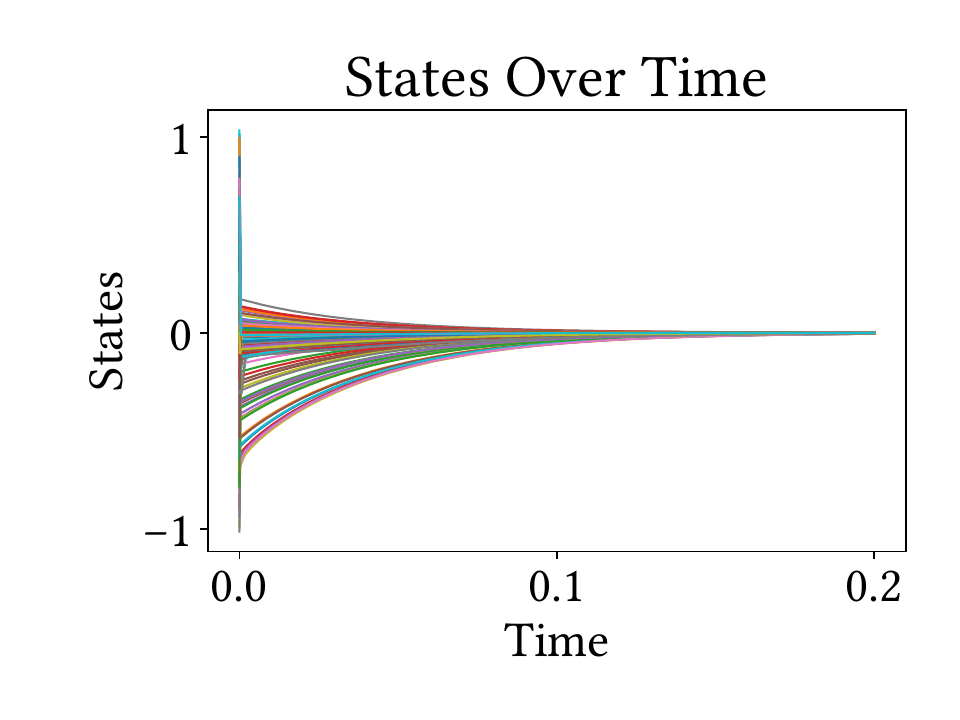}
    \includegraphics[width=0.3\linewidth]{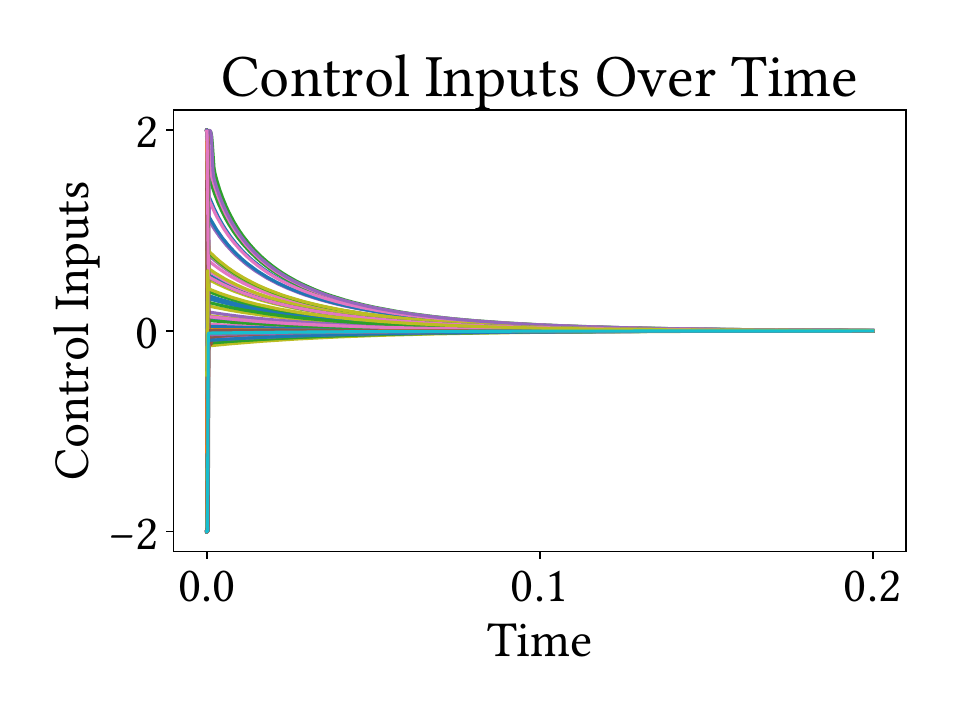}
    \includegraphics[width=0.3\linewidth]{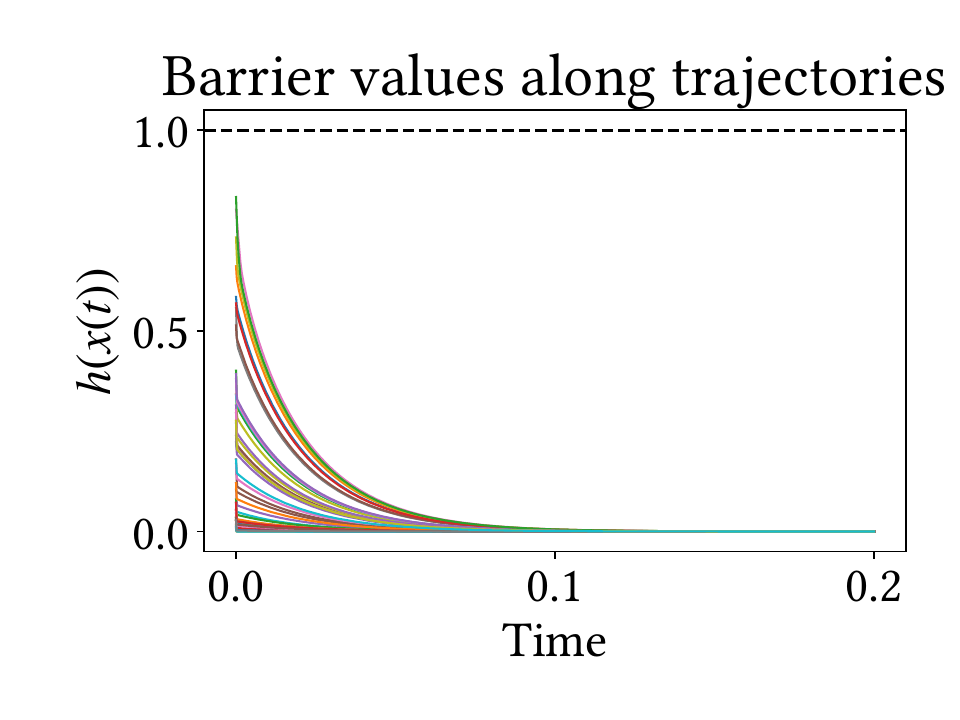}
    \caption{Simulated trajectories, control inputs, and corresponding barrier function evaluations for Example~\ref{ex:3d_power}. All trajectories converge, control input bounds are satisfied, and the barrier function values stay below the safety threshold of~1.}
    \label{fig:3d_power_verified_trajectories}
\end{figure*}

\begin{example}\label{ex:3d_power}
We consider the following model of a power converter, as considered in \cite{dai2024verificationv1}, which is a nonlinear control-affine system~\eqref{eq:sys} with
\[
f(x)=\begin{bmatrix}
-0.05x_1 - 57.9x_2 + 0.00919x_3\\[2pt]
1710x_1 + 314x_3\\[2pt]
-0.271x_1 - 314x_2
\end{bmatrix},
\]
and
\[
g(x)=\begin{bmatrix}
0.05 - 57.9x_2 & -57.9x_3\\[2pt]
1710 + 1710x_1 & 0\\[2pt]
0 & 1710 + 1710x_1
\end{bmatrix},
\]
on the domain $[-2,2]^3$. The unsafe set is specified by
\[
x_1 \le 0.2, 
\quad x_1 \ge -0.8, 
\quad (x_2-0.001)^2 + x_3^2 \le 1.2^2,
\]
together with the box constraints defining the domain. Here we also add an input constraint $\U=[-2,2]^2$. We construct a softmax barrier function $h=h_{\mathrm{sm}}$ with $\tau=3.1$ and verify strict compatibility with a quadratic CLF $V$ (obtained by local LQR) on the boundary band. The patched CLBF obtained from~\eqref{eq:patched_W} using Theorem~\ref{thm:patch} is depicted in Fig.~\ref{fig:3d_power}, and the resulting state and input trajectories under the controller by Leyva \textit{et al.} \cite{leyva2013global} are depicted in Fig.~\ref{fig:3d_power_verified_trajectories}. 
\end{example}

\textbf{Computational time:} All verification was performed using the SMT solver dReal \cite{gao2013dreal}. The computational times are summarized in Table \ref{tab:time} for reference. Compared with the results in \cite{liu2025computing}, and in view of the conditions formulated in Propositions \ref{prop:cbf}--\ref{prop:clbf-compat-bounded-control}, handling control constraints required introducing additional auxiliary variables during verification, which increased the verification time.

\begin{table}[h!]
  \caption{Verification times for numerical examples}
  \label{tab:time}
  \begin{tabular}{ccc}
    \toprule
Model &  CBF (sec) & CLF-CBF compatibility (sec) \\ 
    \midrule
Example \ref{ex:2d_toy_verified} (toy) & 0.07 & 166\\
Example \ref{ex:3d_power} (converter) & 0.05 & 75\\
\bottomrule
\end{tabular}
\end{table}

\section{Conclusions}

We presented a verification framework for synthesizing a single smooth Lyapunov function that certifies both safety and stability under bounded control inputs. The method is based on an explicit smooth patching construction derived from strictly compatible CBF-CLF pairs. Formal guarantees are established using $\delta$-complete SMT solvers, and validation on two nonlinear systems shows that the proposed approach yields less conservative certified regions than sum-of-squares-based alternatives. Future work will focus on scalable synthesis for higher-dimensional systems and compositional verification of interconnected control architectures.

\bibliographystyle{plain}        %
\bibliography{ecc26}

\appendices

\section{Proof of Lemma \ref{lem:conic_farkas}}

\noindent\textbf{Lemma~\ref{lem:conic_farkas} (restated).}
Let $A \in \Real^{n \times m}$, $b \in \Real^n$, and $B>0$.  
The system $A u < b$ has a solution satisfying $\|u\|_2 \le B$ if and only if for every $\lambda \in \Real^n_{\ge 0}$, $\tau \in \Real_{\ge 0}$, and $y \in \Real^m$ satisfying 
$y + A^\top \lambda = 0$, $\|y\|_2 \le \tau$, and $\sum_{i=1}^n \lambda_i + \tau = 1$ also satisfies $\lambda^\top b + B \tau > 0.$

\begin{proof}
\textbf{($\Longrightarrow$)} 
Let $u_\star$ satisfy $A u_\star < b$ and $\|u_\star\|_2 \le B$.
Take any $\lambda \ge 0$, $\tau \ge 0$, $y$ with $y + A^\top\lambda = 0$, $\|y\|_2 \le \tau$, and $\sum_i \lambda_i + \tau = 1$.

\emph{Case 1: $\lambda = 0$.} Then $y=0$ and $\tau=1$, so
\[
\lambda^\top b + B\tau = B > 0.
\]

\emph{Case 2: $\lambda \neq 0$.} Since $b - A u_\star \in \mathbb{R}^n_{>0}$ and $\lambda \in \mathbb{R}^n_{\ge 0}\setminus\{0\}$,
\[
\lambda^\top(b - A u_\star) > 0.
\]
By Cauchy-Schwarz, we have 
\[
y^\top u_\star \le \|y\|_2\,\|u_\star\|_2 \le \tau \|u_\star\|_2.
\]
Using this, $y = -A^\top \lambda$, and $\|u_\star\|\le B$, we have 
\begin{align*}
\lambda^\top b + B\tau & \ge \lambda^\top b + \tau\|u_\star\|_2 \\
& \ge \lambda^\top b + y^\top u_\star = \lambda^\top(b - A u_\star) \;>\; 0.
\end{align*}
Thus in both cases $\lambda^\top b + B\tau > 0$.

\textbf{($\Longleftarrow$)}  Define the convex sets
\[
C := \{(A u - b,\, \|u\|_2) : u \in \mathbb{R}^m\} \subset \mathbb{R}^{n} \times \mathbb{R},
\]
and
\[
K := \mathbb{R}^n_{<0} \times (-\infty, B].
\]
The constraint $\exists u\in\Real^m$ such that $A u < b$ and $\|u\|_2 \le B$ is equivalent to
$C \cap K \neq \emptyset$. We prove the contrapositive. Suppose no $u$ satisfies
$A u < b$ and $\|u\|_2 \le B$. Then $C \cap K = \emptyset$.
By the \emph{strong hyperplane separation theorem}. By the separating hyperplane theorem,
there exist $(\lambda,\tau)\neq(0,0)$ and $\beta\in\mathbb{R}$ such that
\begin{equation}
    \label{eq:K}
\lambda^\top s+\tau t \;\le\; \beta \quad \forall (s,t)\in K,
\end{equation}
and
\begin{equation}
    \label{eq:C}
\lambda^\top (Au-b)+\tau\|u\|_2 \;\ge\; \beta \quad \forall u\in\mathbb{R}^m.
\end{equation}
Because $K=\{(s,t): s\in \mathbb{R}^n_{<0},\ t\le B\}$ is unbounded toward $s\to -\infty$ (componentwise) and $t\to -\infty$, the inequality (\ref{eq:K}) forces
\begin{equation}
    \label{eq:beta_1}
\lambda\ge 0,\quad \tau\ge 0,
\quad
\beta \;\ge\; \sup_{(s,t)\in K}(\lambda^\top s+\tau t) \;=\; B\tau.
\end{equation}
On the other hand, write
\[
\lambda^\top (Au-b)+\tau\|u\|_2
= -\lambda^\top b + (A^\top\lambda)^\top u + \tau\|u\|_2.
\]
For the lower bound $\beta$ in (\ref{eq:C}) to hold for all $u$, we must have
\[
\inf_{u\in\Real^m}\big\{(A^\top\lambda)^\top u + \tau\|u\|_2\big\} \;>\; -\infty
\quad\Longleftrightarrow\quad
\|A^\top\lambda\|_2 \le \tau,
\]
in which case the infimum equals $0$ (attained at $u=0$). Therefore
\begin{equation}
    \label{eq:beta_2}
-\lambda^\top b \;\ge\; \beta.
\end{equation}
Combining (\ref{eq:beta_1}) and (\ref{eq:beta_2}) leads to
\[
\lambda^\top b + B\tau \;\le\; 0.
\]
Set $y:=-A^\top\lambda$ so that $y+A^\top\lambda=0$ and $\|y\|_2=\|A^\top\lambda\|_2\le\tau$.
Normalize by a positive scalar so that $\sum_i \lambda_i+\tau=1$.
This yields a triple $(\lambda, \tau, y)$ satisfying all the conditions in the hypothesis, but for which $\lambda^\top b + B\tau > 0$ fails to hold. Hence, by contrapositive, $C\cap K\neq\emptyset$, i.e., there exists $u\in\Real^m$ with $Au<b$ and $\|u\|_2\le B$.
\end{proof}

\end{document}